\documentclass[conference]{IEEEtran}
\IEEEoverridecommandlockouts
\usepackage{cite}
\usepackage[cmex10]{amsmath}
\usepackage{amssymb,amsfonts}
\usepackage{algorithmic}
\usepackage{graphicx}
\usepackage{textcomp}
\usepackage{xcolor}
\usepackage{amsmath,amsfonts,amsthm,amssymb}
\usepackage{algorithmic,mathrsfs}
\usepackage{algorithm,dutchcal}
\theoremstyle{remark}
\newtheorem{theorem}{\hskip 1em Theorem}
\newtheorem{lemma}{\hskip 1em Lemma}
\newtheorem{remark}{\hskip 1em Remark}
\newtheorem{definition}{\hskip 1em Definition}
\newtheorem{proof skecth}{Proof skecth}
\usepackage{array,caption}
\usepackage[mathscr]{eucal}
\usepackage[caption=false,font=LARGE,labelfont=rm,textfont=rm]{subfig}
\usepackage{textcomp}
\usepackage{stfloats}
\usepackage{url}
\usepackage{verbatim}
\usepackage{graphicx,float,balance}
\usepackage{multirow}
\usepackage{cite}
\usepackage{soul}
\usepackage{color,xcolor}
\def\BibTeX{{\rm B\kern-.05em{\sc i\kern-.025em b}\kern-.08em
		T\kern-.1667em\lower.7ex\hbox{E}\kern-.125emX}}
\begin{document}
	\title{Second-Order Identification Capacity of AWGN Channels}
	
	\author{
	\IEEEauthorblockN{Zhicheng Liu\IEEEauthorrefmark{1}\IEEEauthorrefmark{2}, Yuan Li\IEEEauthorrefmark{2},
		Huazi Zhang\IEEEauthorrefmark{2}, Jun Wang\IEEEauthorrefmark{2}, Guiying Yan \IEEEauthorrefmark{3}\IEEEauthorrefmark{4} and Zhiming Ma \IEEEauthorrefmark{3}\IEEEauthorrefmark{4}}
	\IEEEauthorblockA{\IEEEauthorrefmark{1}%
		School of Mathematics and Statistics, Beijing Jiaotong University, Beijing, China, 20118002@bjtu.edu.cn}
	\IEEEauthorblockA{\IEEEauthorrefmark{2}%
		Huawei Technologies Co. Ltd., Hangzhou China, \{liyuan299, zhanghuazi, justin.wangjun\}@huawei.com}
	\IEEEauthorblockA{\IEEEauthorrefmark{3}%
		University of Chinese Academy and Sciences}
	\IEEEauthorblockA{\IEEEauthorrefmark{4}%
		Academy of Mathematics and Systems Science, CAS, Beijing China, yangy@amss.ac.cn, 	mazm@amt.ac.cn
	}
}
	
	\maketitle
	\begin{abstract}
		In this paper, we establish the second-order randomized identification capacity (RID capacity) of the Additive White Gaussian Noise Channel (AWGNC). On the one hand, we obtain a refined version of Hayashi's theorem to prove the achievability part. On the other, we investigate the relationship between identification and channel resolvability, then we propose a finer quantization method to prove the converse part. Consequently, the second-order RID capacity of the AWGNC has the same form as the second-order transmission capacity. The only difference is that the maximum number of messages in RID \emph{scales double exponentially} in the blocklength.
	\end{abstract}
	
	\begin{IEEEkeywords}
		randomized identification, AWGN channels, channel resolvability, quantization.
	\end{IEEEkeywords}
	
	\section{Introduction}
	\IEEEPARstart{W}{ith} the emergence of IoT applications \cite{ref1,ref2}, modern communications in the next-generation wireless network framework (XG) require robust and ultra-reliable low latency information exchange between a large pool of potential smart devices. Numerous XG applications	are event-triggered communication systems, such as vehicle-X communication \cite{ref3,ref4,ref5},  tactile internet \cite{ref6, ref7}, industry 4.0, Online sales, etc. The emergence of novel communication tasks, including control systems \cite{ref8}, the automotive domain \cite{ref9}, watermarking \cite{ref10,ref11,ref12}, recommendation systems \cite{ref13}, and other scenarios requiring quick or small checks, presents new challenges. 
	For many of these problems, signaling a user, rather than transmitting a bulk of data, becomes the main task. As such, the identification approach suggested by Ahlswede and Dueck \cite{ref14} is more suitable than the transmission scheme as studied by Shannon \cite{ref15}.  Some other possible applications of identification codes have been pointed out in \cite{ref31,ref27}.
	
	In identification, receiver-$i$ only cares whether message-$i$ is transmitted. Once receiver-$i$ believes message-$i$ is not sent, it does not attempt to decode that message. %Possible applications for identification include authentication tasks such as watermarking [5-8], as well as event-driven applications encountered in Industry 4.0 and vehicle-X communiction [9]. 
	In randomized identification, message-$i$ is sent by randomly transmitting a codeword from receiver-$i$’s codebook, and it was proved that the optimal code size scales double exponentially in the block length with vanishing type-I and type-II error rates. Han and Verd$\mathrm{\acute{u}}$ \cite{ref16} provided a new perspective named channel resolvability to further discuss the identification capacity. Subsequently, Steinberg \cite{ref17} proposed a nuch more tighter converse bound for the randomized identification capacity. Hayashi \cite{ref18} extended the previous results to wiretap channels and established the error exponents. Han obtained the identification capacity of the continuous input channels in \cite{ref19}. The second-order RID capacity of the discrete memoryless channels (DMCs) was  demonstrated in \cite{ref20} which relies on the finiteness assumption of the input alphabet, hence can not apply to continuous input channels directly. Based on constant-weight codes (CWCs) that result from concatenating a CWC initialization with outer linear block codes, Verd$\mathrm{\acute{u}}$ and Wei \cite{ref21} suggested the first explicit ID code construction.
	
	In the deterministic setup, given a discrete memoryless channel (DMC), the number of messages grows exponentially with the blocklength. R. Ahlswede and Ning Cai determined the deterministic identification capacity (DID capacity) for DMCs \cite{ref22}.  In particular, J$\mathrm{\acute{a}}$J$\mathrm{\acute{a}}$ \cite{ref23} showed that the deterministic identification (DI) capacity of a binary symmetric channel is $1$ bit per channel use, as one can exhaust	the entire input space and assign (almost) all binary $n$-tuples as codewords. Recently, Salariseddigh et al. established the deterministic identification capacity for AWGN channels and AWGN channels with slow and fast fading coefficients \cite{ref24, ref25}. Li et al.  \cite{ref26} proposed the deterministic identification capacity for block memoryless fading channels without CSI.
	%In the rest of this paper, we mainly focus on the second-order RID capacity of the AWGNC. 
	
	Figure $\ref{fig1}$ provides a brief overview of the connections between our work and previous research. We extend the results presented by Watanabe to the continuous input alphabet, specifically, by determining the second-order RID capacity for the AWGNCs. 
	\begin{figure}[H]
		\centering
		\includegraphics[height=3.5cm,width=8.5cm]{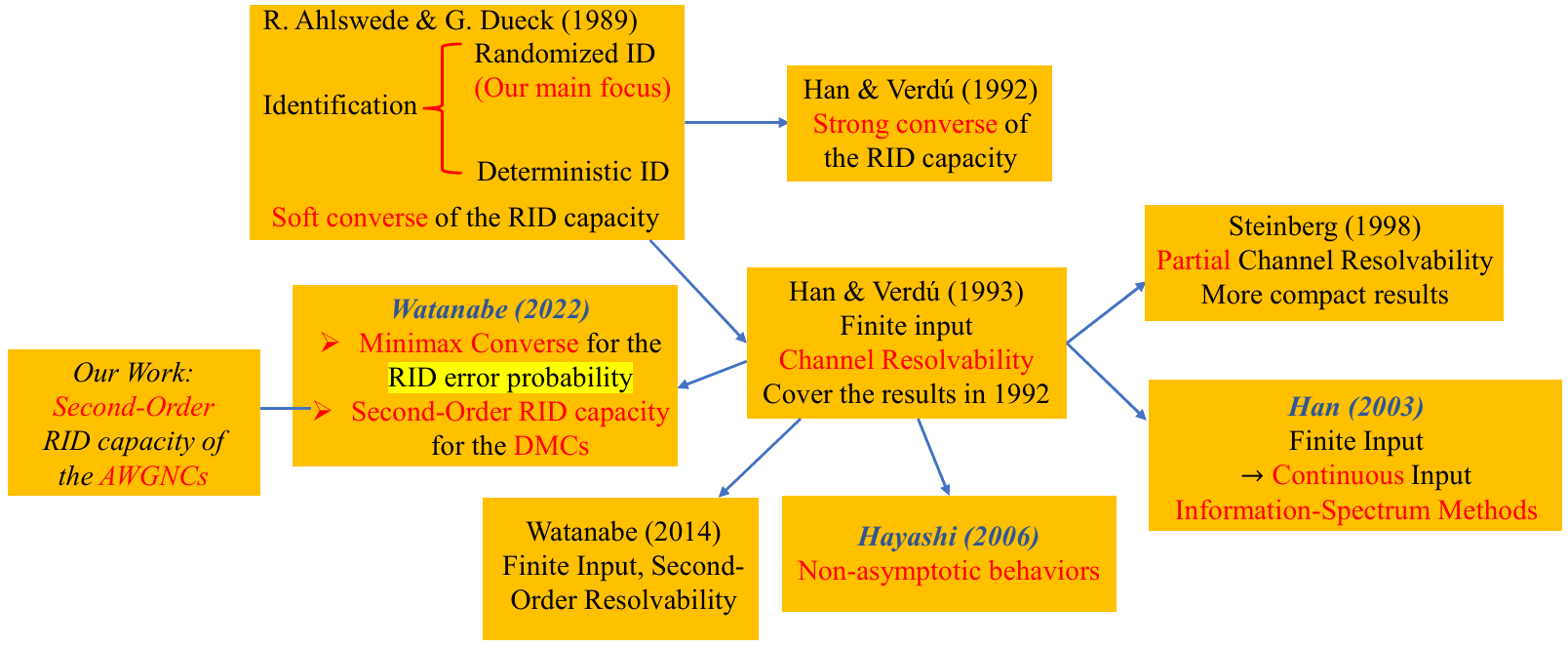}
	\captionsetup{justification=centering}
		\caption{Research status of the Randomized Identification Problems.}
		\label{fig1}
	\end{figure}
	
	In this paper, we derive the second-order RID capacity of the AWGNC. Our contributions are summarized as follows:
	\begin{enumerate}
		\item{We propose a refined version of Hayashi's theorem \cite{ref18}, which results in the achievability part.}
		\item{We propose a finer quantization method than that in \cite{ref19} to obtain the converse part.}
	\end{enumerate}
	
	The remainder of this paper is structured as follows. Section II provides the background of the identification and the resolvability. Section III introduces the main results of the paper: the second-order RID capacity of the AWGNC. Finally, we draw conclusions in Section IV. The main differences between our work and the previous studies are provided in Table \ref{tab_demo}.
	%%%% recalls some basic definitions related to RID and resolvability
	
	\textit{Notation Conventions:} In this paper, the lowercase letters $x$,$y$,$z$,$\cdots$ represent the value of random variables, and uppercase letters $X$,$Y$,$Z$,$\cdots$ represent random variables. The probability distribution of random variable $X$ is specified by a cumulative distribution function (CDF) $F_{X}(x) = \mathrm{Pr}(X \leq x)$, or alternatively, by a probability density function (PDF) $P_{X}(x)$. 
	$X^{n}=(X_{1},...,X_{n})$ and $x^{n} = (x_{1}, x_{2}, . . . , x_{n})$ are  random vectors of length-$n$ and its realization, respectively. $\mathcal{X}^{n}$ and $\mathcal{Y}^{n}$ represent the $n$th Cartesian product of input alphabet and output alphabet. The $\ell^{2}$-norm of $x^{n}$ is denoted by $||x^{n}||$.
	A probability distribution $P_{X^{n}}$ is called an $M$-type, if for any integer $M>0$ and every $x^{n}\in\mathcal{X}^{n}$, 
	\begin{flalign*}
		P_{X^{n}}(x^{n})\in\{0,\frac{1}{M},\frac{2}{M},\cdots,1\}.
	\end{flalign*}
	The variational distance between two probability distributions $P_{X_{1}}$ and $P_{X_{2}}$ defined on the same measurable space $(\Omega,\mathcal{F})$ is $d(P_{X_{1}},P_{X_{2}})=2\sup_{E}|P_{X_{1}}(E)-P_{X_{2}}(E)|$. Denote $\mathcal{P}(\mathcal{X}^{n})$ the set of all distributions supported on $\mathcal{X}^{n}$. $\log x$ is the base-2 logarithm of $ x$ and $\mathrm{I}_{n}$ is the $n\times n$ identity matrix.%All information quantities and rate are based on the natural logarithm.
	\section{Background}
	In this section, we introduce the identification and resolvability in AWGNC, refer to \cite{ref14,ref16,ref17,ref18,ref19,ref20} for the results in DMCs. 
	\subsection{Identification via Channels}
	In this section, we review some basic results of randomized identification. %Given a channel $W$ from $\mathcal{X}$ to $\mathcal{Y}$, the sender tries to transmit one of N messages while the receiver needs to identify whether the message $i\in \{1,...,N\}$ is transmitted. 
	In AWGNC, the received sequence $Y^{n}=X^{n}+Z^{n}$, where $Z^{n}\sim\mathscr{N}(0,\sigma^{2}\mathrm{I}_{n})$. Denote the signal-to-noise ratio by $\mathrm{SNR}=\frac{P}{\sigma^{2}}$, and without loss of generality, we assume $\sigma^{2}=1$. The channel $W^{n}$ is $W^{n}(y^{n}|x^{n})=(2\pi)^{-\frac{n}{2}}e^{-||y^{n}-x^{n}||^{2}/2}$. The receiver shall identify if message $i\in\{1,\cdots,N\}$ was transmitted or not. A randomized identification code is defined by a family  $\{(Q_{i}, \mathcal{D}_{i})_{i=1}^{N}\}$, where $Q_{i}\in \mathcal{P}(\mathcal{X}^{n}(P)) (i=1,\cdots,N)$ are randomized encoders with $\mathcal{X}^{n}(P)=\{x^{n}\in\mathcal{X}^{n}\big| ||x^{n}||^{2}\leq nP\}$ and $P$ is the power constraint,
	$\mathcal{D}_{i}\subset \mathbb{R}^{n} (i=1,\cdots,N)$ are decoding regions. %considers message-$i$ as the target message if and only if $y^{n} \in \mathcal{D}_{i}$.
	
	\begin{table*}
		\caption{Comparisons between this work and the previous stuidies.}
		\label{tab_demo}
		\centering	
		\tabcolsep=0.001\linewidth
		%	\resizebox{1.0\linewidth}{!}{
			%	
			%}
		\begin{tabular}{ccc}
			\hline
			& Properties & Results \\
			\hline
			\multirow{2}{*}{R. Ahlswede \& G. Dueck (1989) \cite{ref14}}  & \multirow{2}{*}{Finite Input Alphabets, exponentially decay error probability} & \multirow{2}{*}{Soft Converse}  \\
			&  &  \\
			\multirow{2}{*}{T. S. Han \& Verd$\mathrm{\acute{u}}$ (1993) \cite{ref16}}  & \multirow{2}{*}{Finite Input \& Strong Converse Property} & \multirow{2}{*}{Resolvability \& Strong Converse} \\
			%		\multirow{4}{*}{T. S. Han \& Verd$\mathrm{\acute{u}}$ (1993) [14]}  & \multirow{2}{*}{Identification $\rightarrow$ Channel Resolvability} & Identification Converse \\
			%		& &  \\
			& & \\
			\multirow{2}{*}{Y. Steinberg (1998) \cite{ref17}} & \multirow{2}{*}{Finite Input \& Partial Channel Resolvability} & \multirow{2}{*}{Tighter Identification converse bound}    \\
			& &\\
			\multirow{2}{*}{T. S. Han (2003) \cite{ref19}} & \multirow{2}{*}{Finite Input Alphabet $\rightarrow$ Continuous Input Alphabet}  & \multirow{2}{*}{$\log\log N^{*}(\varepsilon, \delta|W^{n}) = n C(P)$}    \\
			& &\\
			\multirow{3}{*}{M. Hayashi (2006) \cite{ref18}} & \multirow{3}{*}{Finite Input \& Wiretap Channels}  & \multirow{3}{*}{ID Converse \& Error Exponents}   \\
			\multirow{4}{*}{} &\multirow{4}{*}{}  & \multirow{4}{*}{} \\
			&&\\
			S. Watatnabe (2022) \cite{ref20}& Minimax Converse for Identification via arbitrary finite input channels & $\log\log N^{*}(\varepsilon, \delta|W^{n})=$\\
			& \& Second-Order Identification Capacity of the DMCs, $\delta\rightarrow 0$& $n\mathrm{C(W)}-\sqrt{n\mathrm{V(\varepsilon|W)}}\mathrm{Q}^{-1}(\varepsilon)+\mathrm{o}(\sqrt{n})$\\
			\multirow{2}{*}{} & \multirow{2}{*}{} & \multirow{2}{*}{} \\
			This work & Second-order Identification Capacity of the AWGN Channels, $\delta \rightarrow 0$  &  Theorem \ref{th1}, $\log\log N^{*}(\varepsilon,\delta|W^{n}) =$  \\
			& & $n\mathrm{C}(P)$\textcolor{red}{$-\sqrt{n\mathrm{V}(P)}\mathrm{Q^{-1}(\varepsilon)}+\mathrm{O}(\log n)$} \\
			\hline
		\end{tabular}
	\end{table*}
	For a given ID code ${(Q_{i},\mathcal{D}_{i})}_{i=1}^{N}$, we define the type-I error (missed detection error) and type-II error (false activation rate) %(see Fig $\ref{fig2}$ for an illustration\footnote{Here, the codewords $c_{i}\sim Q_{i} (i=1,\cdots,6)$, where $Q_{i}$s are the randomized encoders.} 
	as
	\begin{flalign}
		\mathrm{P_{I}} \triangleq \max_{1\leq i\leq N} Q_{i}W^{n}(\mathcal{D}^{c}_{i}),
		\label{eq1}\\
		\mathrm{P_{II}} \triangleq \max_{1\leq j\neq i\leq N} Q_{i}W^{n}(\mathcal{D}_{j}),
		\label{eq2}
	\end{flalign}
	where $Q_{i}W^{n}$ is the output distribution induced by the input distribution $Q_{i}$, i.e.,
	\begin{flalign}
		Q_{i}W^{n}(y^{n})=\sum_{x^{n}\in\mathcal{X}^{n}(P)}Q_{i}(x^{n})W^{n}(y^{n}|x^{n}).
		\label{eq3}
	\end{flalign}
%	\begin{figure}[H]
%		\centering
%		\includegraphics[height=4cm,width=7.5cm]{ID_error.png}
%		\captionsetup{justification=centering}
%		\caption{Illustration of the type-I and type-II errors [11].}
%		\label{fig2}
%	\end{figure}
	Note that the codewords $c_{i}\sim Q_{i}(i=1,\cdots,6)$. 
	
	For all $0< \varepsilon,\ \delta <1$ with $0<\varepsilon+\delta<1$, an ID code ${(Q_{i},\mathcal{D}_{i})}_{i=1}^{N}$ is an $(n,N,\varepsilon,\delta)$-ID code if $\mathrm{P_{I}}\leq \varepsilon$ and $\mathrm{P_{II}}\leq \delta$. 
	\\
	
	Denote $N^{*}(\varepsilon,\delta|W^{n})$ the optimal code size, i.e.,
	\begin{flalign}
		N^{*}(\varepsilon,\delta |W^{n})\triangleq\sup\left\lbrace N\big|(n,N,\varepsilon,\delta)\mathrm{-ID\ code \ exists}\right\rbrace.
		\label{eq4}
	\end{flalign} 
	%\begin{figure}[H]
	%	\centering
	%	\includegraphics[height=4.2cm,width=7.5cm]{ID_error.png}
	%	%\captionsetup{justification=centering}
	%	\caption{\fontsize{12bp}{17bp}Illustration of the type-I and type-II errors of identification via channels [11].}
	%	\label{fig2}
	%\end{figure}
	\quad Since the decoding regions can be overlapped in randomized identification while must be disjoint in transmission, the RID allows much more messages. Furthermore, it is necessary to emphasize that random coding is one of the main factors responsible for the performance gain. Shannon \cite{ref15} proved that the optimal message number scales exponentially with $n$, and provided the transmission capacity \emph{in the exponential scale}, i.e., the channel capacity. %In transmission, the optimal message number scales exponentially with the code length, and the transmission capacity in the \emph{exponential scale} was given by Shannon \cite{ref7}.
	While in RID, the capacity in the \emph{exponential scale} is infinite, Ahlswede and Dueck \cite{ref14} further pointed out that the optimal message number scales double exponentially with $n$ for DMCs, they denoted 
	$$\mathrm{C_{ID}}(\varepsilon,\delta|W^{n})=\lim_{n\rightarrow \infty}\frac{1}{n}\log\log N^{*}(\varepsilon,\delta|W^{n})$$
	the RID capacity in the \textit{double-exponential scale}. Surprisingly, $\mathrm{C_{ID}}(\varepsilon,\delta|W^{n})$ equals the Shannon capacity in DMCs. The result was extended to AWGNC by Han \cite{ref19}, i.e.,
	\begin{flalign}
		\mathrm{C_{ID}}(\varepsilon,\delta|W^{n})= \mathrm{C}(P) = \frac{1}{2}\log(1+P),
		\label{eq5}
	\end{flalign}
	as long as $0<\varepsilon+\delta<1$. Notice that the RID capacity and the channel capacity are only numerically equal, while the meanings of the two are different. In fact, the number of messages supported by the RID  far exceeds that in transmission!
	
	Additionally, in \cite{ref27}, the identification code for AWGN is analyzed without randomized encoding. The performance remains superior in terms of transmission, yet with $N^{*}(\varepsilon,\delta|W^{n})=2^{R n\log n}$ where $R$ is the optimal rate (\textit{super-exponential} but not \textit{double-exponential}!).
	
	However, the above first-order approximation (\ref{eq5}) is only valid when $n$ is sufficiently large, thus we still need to propose more accurate expressions to estimate the performance of identification in the finite length regime. It is worth noting that there exists a constructive proof of finite-length identification codes \cite{ref28} with a code construction which was analyzed in \cite{ref29} and implemented in \cite{ref13}. 
	
	In the next section, we provide our main results, i.e., the second-order RID capacity of the AWGNC. In order to prove it, we need to introduce the channel resolvability \cite{ref16} which plays a major role in the proof of the converse bound.
	
	\subsection{Channel Resolvability}
	Given a target distribution $P_{Y^{n}}$, the channel resolvability minimizes the size $M$ of the codebook $c^{M}=(c_{1},...,c_{M})$ such that when the codewords are equiprobably selected, the output distribution
	\begin{flalign}
		P_{Y^{n}[c^{M}]} =\frac{1}{M}\sum_{i=1}^{M}W(Y^{n}|c_{i})
		\label{eq6}
	\end{flalign}
	approximates the target distribution $P_{Y^{n}}$ well, apparently, the output distribution $P_{Y^{n}[c^{M}]}$ is an $M$-type. Figure \ref{fig2} shows the main idea of the channel resolvability framework.
	\begin{figure}[H]
		\centering
		\includegraphics[height=2.3cm,width=8.3cm]{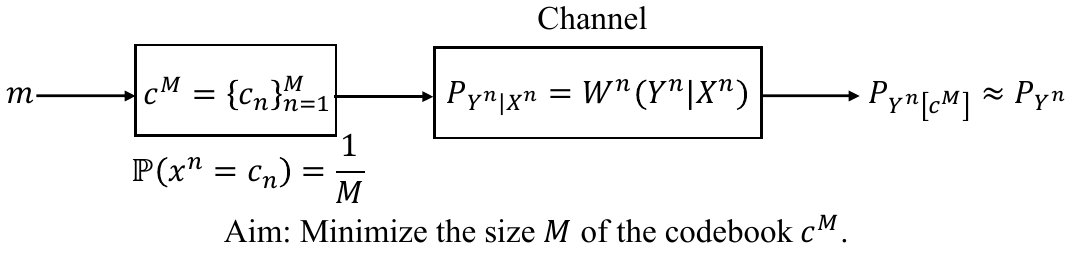}
		\caption{Problem formulation of the channel resolvability.}
		\label{fig2}
	\end{figure}
	\begin{definition}
		A resolvability code %$c^{M}=(c_{1},\cdots,c_{M})$ 
		for input distribution $P_{X^{n}}$ is called an $(n,M,\xi)$-resolvability code if
		%	Rate $R$ is $\xi$-achievable for an input $\mathbf{X}=\{X^{n}\}^{\infty}_{n=1}$ if and only if there exists an input $\{\tilde{X}^{n}\}_{n=1}^{\infty}$ satisfying 
		\begin{flalign}
			\limsup_{n\rightarrow \infty}d(P_{Y^{n}},P_{{Y}^{n}[c^{M}]})\leq \xi, \tag{a}
		\end{flalign}
		where $P_{Y^{n}}$ denotes the channel output induced by $P_{X^{n}}$.
	\end{definition}

	%A resolvability code is called $(M,\xi)$-resolvability code if $(*)$ is satisfied. 
	The optimal code size %$M^{*}(\xi|W^{n})$ 
	of resolvability is defined by
	\begin{flalign}
		& M^{*}(\xi|W^{n})\triangleq \inf\{M\big| (n,M,\xi) \mathrm{-resolvability\ code\ exists} \nonumber\\
		& \mathrm{\ for\  all\ input\ distributions\ } P_{X^{n}} \in\mathcal{P}(\mathcal{X}^{n}(P))  \}.\tag{b}
	\end{flalign}
	
	The channel resolvability capacity of the AWGNC is defined as $\mathrm{C_{RE}(\xi)}=\lim_{n\rightarrow \infty}\frac{1}{n}\log M^{*}(\xi|W^{n}).$
	
	It has been proved by Han and Verd$\acute{u}$ that the converse bound of the RID capacity could be upper bounded by the achievability bound of the channel resolvability capacity (see Figure \ref{fig3})
	\begin{figure} 
		\centering       
		\includegraphics[height=3.5cm,width=6.8cm]{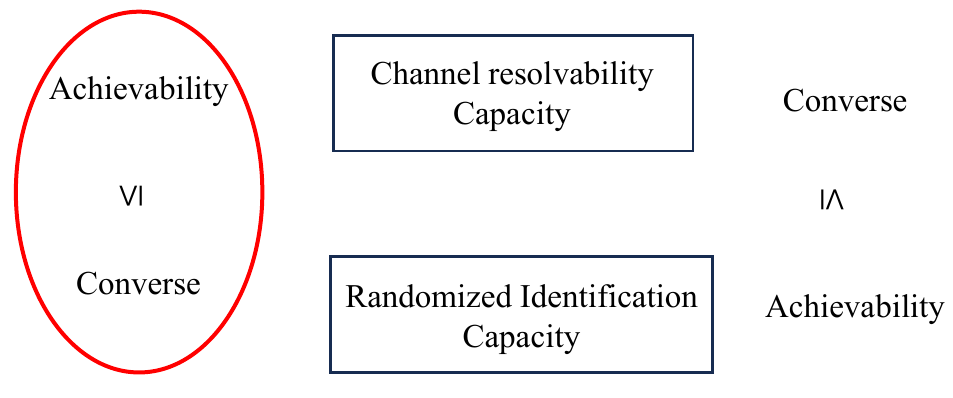}
		\caption{The dual relation between the identification and the resolvability.}   
		\label{fig3}     
	\end{figure}
	
	It is worth mentioning that Han \cite{ref19} pointed out that the channel resolvability capacity equals to the identification capacity of the AWGNC as long as $0<\varepsilon+\delta+\xi<1$. A substitute approach for Han and Verd$\mathrm{\acute{u}}$'s resolvability was subsequently proposed by Ahlswede in a more combinatorial manner. This concept is elaborated upon in \cite{ref30}.
	\subsection{Quantization}
	The minimax converse of identification was given by Watanabe \cite{ref20} for DMCs. The proof relies on the finiteness assumption of the input alphabet. Nevertheless, it becomes impossible to obtain an upper bound on the number of distinct types for continuous input alphabet, as in AWGNC. Therefore, we use quantization method to overcome this difficulty. 
	
	In \cite{ref19}, Han proposed the quantization method to derive the first-order identification capacity $(\ref{eq5})$ in AWGNC. Fig $\ref{fig4}$ shows the basic ideas of Han's quantization method. An alternative to Han's quantization method is suggested in \cite{ref31}, and this approach is employed in \cite{ref28} to determine the capacity of secure identification codes.
	
	Specifically, the input alphabet is % the power shell 
	$\mathcal{X}^{n}(P)=\{x^{n}\big| ||x^{n}||^{2}\leq nP\}$. % , where $P$ is the power constraint. 
	Find the minimal hypercube $V_{n}(P)$ with edge length $l_{n}=2\sqrt{nP}$ covering $\mathcal{X}^{n}(P)$ and partition $V_{n}(P)$ into small hypercubes $\Lambda_{n}^{(i)}$ with edge length $\Delta_{n}$ $(i=1,\cdots,k_{n}(P))$. The number of small hypercubes is  $k_{n}(P)=(\frac{l_{n}}{\Delta_{n}})^{n}$. We choose a representative point $u^{n}_{i}$ in each $\Lambda^{(i)}_{n}$ and set 
	\begin{flalign*}
		\mathcal{R}_{n}(P)=\{u^{n}_{1},\cdots,u^{n}_{k_{n}(P)}\}.
	\end{flalign*}
	\begin{figure}
		\centering
		\includegraphics[height=5.3cm,width=7cm]{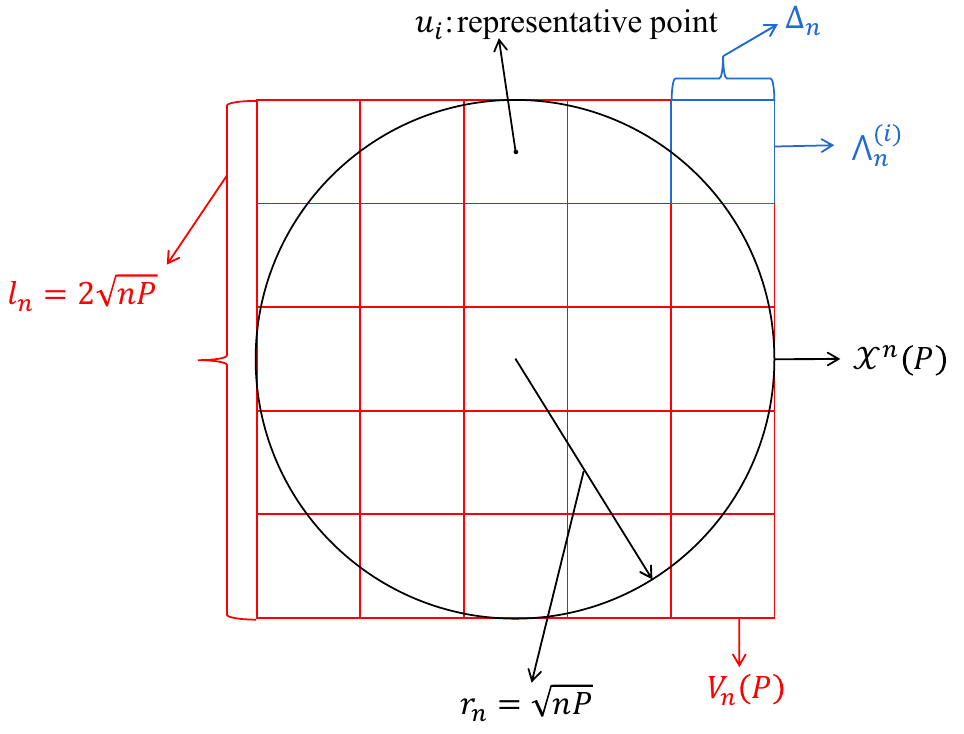}
		\caption{Intuition behind the quantization method.}
		\label{fig4}
	\end{figure}
	For a distribution $Q$ supported on $\mathcal{X}^{n}(P)$, we denote $\bar{Q}$ supported on $\mathcal{R}^{n}(P)$ by
	\begin{equation}
		\bar{Q}(u^{n}_{i})=Q(\Lambda^{(i)}_{n})  (i=1,2,...,k_{n}(P)).
		\label{eq7}
	\end{equation}
	%i.e., $\bar{Q}$ is obtained by concentrating the mass of $Q$ in $\Lambda_{n}^{(i)}$ to the representative point $u^{n}_{i}$. %mass in $\Lambda_{n}^{(i)}$ is concentrated to the representative point $\mathbf{u}_{i}$.
	
	When the edge length $\Delta_{n}$ is small enough, the output distribution 
	induced by $\bar{Q}$ approximates the real output distribution induced by $Q$ (see (\ref{eq3})) well. Therefore, the problem boils down to the finite input alphabet $\mathcal{R}_{n}(P)$, which can be solved as that in DMCs.
	
	\section{Second-Order RID capacity of the AWGNC}
	In this section, we obtain the second-order RID capacity of the AWGNC (Theorem 1). In the achievability part, we choose the uniform distribution on the power shell as the input distribution $P_{X^{n}}(x^{n})$, so the components of the output distribution $P_{Y^{n}}(y^{n})$ induced by $P_{X^{n}}(x^{n})$ are not \emph{i.i.d.}, thus is complex to analysis. In order to overcome this difficulty, we extend Hayashi's Theorem \cite{ref18} by utilizing capacity-achieving output distribution $Q_{Y^{n}}^{*}(y^{n})=\mathscr{N}(0,(1+P)\mathrm{I_{n}})$ as the auxiliary
	output distribution instead of $P_{Y^{n}}$. In the converse part, the key idea is to quantize the input alphabet and then convert the problem to a DMC scenario, which is solved by Watanabe \cite{ref20}. This procedure plays an important role in the derivation of the first-order RID capacity of the AWGNC. In order to get the second-order RID capacity, we directly partition the power shell into small sectors. Compared to the partitions on hypercubes in \cite{ref19}, our quantization method is more accurate.
	%and gives the $\mathrm{O}(\log n)$ term in the  of the AWGNC.

	Let $\mathrm{Q^{-1}}(\cdot)$ be the inverse of the complementary CDF of the standard normal distribution. In \cite{ref20}, the second-order RID capacity of the DMCs is proved to have the same form as the transmission capacity. It was shown by Polyanskiy et.al. \cite{ref32} that the second-order transmission capacity of the AWGNC is%can be characterized as
	\begin{flalign*}
		\log \mathrm{M}^{*}_{\mathrm{T}}(\epsilon|W^{n})=n\mathrm{C}(P)-\sqrt{n\mathrm{V}(P)}\mathrm{Q}^{-1}(\epsilon)+\mathrm{O}(\log n),%\nonumber
	\end{flalign*}
	where $\mathrm{M}^{*}_{\mathrm{T}}(\epsilon|W^{n})$ is the optimal code size of the transmission code with the maximal block error rate less than $\epsilon$ and $\mathrm{V}(P)=\log^{2}e\frac{P(P+2)}{2(P+1)^{2}}$ is the channel dispersion.% of the AWGNC.
	
	We would naturally ask: whether the second-order RID capacity and transmission capacity still have the same form in AWGNC ? Our answer is ``yes''.
	
	\noindent \begin{theorem} \label{th1} Given an AWGNC $W^{n}$, when the type-II error $\mathrm{\delta}$ vanishes, $\forall\ 0<\varepsilon<1$, the optimal code size of RID (defined in $(\ref{eq4})$) is
		\begin{flalign}
			\log\log N^{*}(\varepsilon,\delta|W^{n}) = & n\mathrm{C}(P)-\sqrt{n\mathrm{V}(P)}\mathrm{Q^{-1}(\varepsilon)} \nonumber\\
			&+\mathrm{O}(\log n ),\qquad \delta\rightarrow 0.
			\label{eq8}
		\end{flalign}
	\end{theorem}
	Compared with the first-order approximation (\ref{eq5}) given by Han \cite{ref19}, %i.e., $\log\log N^{*}(\varepsilon,\delta|W^{n})=n\mathrm{C}(P)+\mathrm{o}(n)$, 
	we provide a more accurate approximation which gives the exact second-order term in the asymptotic expansion. We divide the proof into two parts: the achievability part and the converse part.
	\\
	
	\noindent\emph{Proof of achievability.} The key point is the extension of Hayashi's Theorem 1 \cite{ref18} by utilizing the auxiliary
	output distribution. The achievability part is a result of Lemma 1.
	\noindent \begin{lemma}Given an arbitrary channel $W^{n}$ and input distribution $P_{X^{n}}$, denote by $P_{Y^{n}}$ the output distribution induced by $P_{X^{n}}$. Assume that real numbers $c, c', d, d', \tau, \zeta > 0$ satisfy
		\begin{flalign}
			\zeta\log(\frac{1}{\tau}&-1)> \log2+1, 0<\tau<\frac{1}{3}, 0<\zeta< 1,\nonumber\\
			&1>\frac{1}{c}+\frac{1}{c'}, f:=1-\frac{1}{d}-\frac{1}{d'}>0.\nonumber
		\end{flalign}
		Then, for any integer $M>0$ and real number $K>0$, there exists an $(N,\varepsilon,\delta)$-$\mathrm{ID}$ code such that
		\begin{flalign}
			\varepsilon&\leq cd\mathrm{Pr}(\tilde{i}(X;Y)\leq\log K),\\%+a'b'\frac{1}{K}\lceil\frac{M}{c}\rceil<1,\\
			&\delta\leq\zeta+c'd'\frac{1}{K}\lceil \frac{M}{f} \rceil,\\
			&N=\lfloor\frac{e^{\tau M}}{Me}\rfloor 
			\label{eq9,10,11}
		\end{flalign}
		provided that
		\begin{flalign}
			cd\mathrm{Pr}(\tilde{i}(X;Y)\leq\log K)+c'd'\frac{1}{K}\lceil\frac{M}{f}\rceil<1,
			\label{eq12}
		\end{flalign}
		where $\tilde{i}(X;Y)\triangleq\log\frac{W(Y|X)}{Q_{Y}(Y)}$ and $Q_{Y}(Y)$ is an auxiliary distribution.
	\end{lemma}
	%		\noindent\emph{Proof.}
	%			The proof of Lemma 2 is similar to Hayashi's Theorem 1 [7] by replacing the information density $i(X;Y)$ with the modified information density $\tilde{i}(X;Y)$.\qed  
	
	In the following, we choose appropriate parameters in Lemma 1 to prove the achievability part.% of the second-order RID capacity.
	
	Let $W^{n}$ be the length-$n$ AWGNC.  %	 Choose the uniform distribution on the power shell as the input distribution, i.e., $P_{X^{n}}(x^{n})=\frac{\delta(||x^{n}||-\sqrt{n\mathrm{SNR}})}{S_{n}(\sqrt{n\mathrm{SNR}})}$, where $\delta(\cdot)$ denotes the Dirac function and $S_{n}(r^{n})=\frac{2\pi^{n/2}r^{n-1}}{\Gamma(n/2)}$ is the surface area of an $n$-dimensional sphere of radius $r$. 
	Choose $Q_{Y^{n}}^{*}$ as the auxiliary output distribution, namely,
	\begin{flalign}
		Q_{Y^{n}}^{*}(y^{n})=%\prod_{i=1}^{n}Q^{*}_{Y}(y_{i})=
		\mathscr{N}(0,(1+P)\mathrm{I_{n}}).
		\label{eq13}
	\end{flalign}

	Applying	the Central Limit Theorem \cite{ref33}, we have
	\begin{flalign}
		\frac{\tilde{i}(x^{n};Y^{n})-n\mathrm{C}(P)}{\sqrt{n\mathrm{V}(P)}}&=\frac{\log e}{2\sqrt{n\mathrm{V}(P)}}[\frac{||Y^{n}||^{2}}{1+P}-||Z^{n}||^{2}]\nonumber\\
		&\Rightarrow\mathscr{N}(0,1),
		\label{eq14}
	\end{flalign}
	where $\tilde{i}(x^{n};Y^{n})=\log \frac{W^{n}(Y^{n}|x^{n})}{Q_{Y^{n}}^{*}(Y^{n})}$ and ``$\Rightarrow$'' means weak convergence.
	
	Set $c=d=1+\frac{2}{n}, c'=d'=n+2, \tau=\frac{1}{n+2}, \zeta=\frac{1+\log2}{\log n}$. For $K>0$, we apply Lemma 1 by setting $M=\lceil\frac{K}{(n+2)^{4}}\rceil$, then, there exist a constant $F>0$ and a sequence of $(N^{*}(\varepsilon,\delta|W^{n}),\varepsilon_{n},\delta_{n})$-ID codes such that
	\begin{flalign}
		&\log\log N^{*}(\varepsilon,\delta|W^{n})\geq \log K-F\log n ,\nonumber\\
		&\varepsilon_{n}\leq (1+\frac{2}{n})^{2}\mathrm{Pr}(\tilde{i}(x^{n};Y^{n})\leq \log K),\nonumber\\
		&\delta_{n}\leq \frac{1+\log 2}{\log n}+\frac{2}{n+2},\nonumber
	\end{flalign}
	if
	\begin{flalign}
		(1+\frac{2}{n})^{2}\mathrm{Pr}(\tilde{i}(x^{n};Y^{n})\leq \log K)+\frac{2}{n+2}\leq 1.\nonumber
	\end{flalign}
	
\begin{lemma}
	 (Berry-Esseen Theorem for Functions of i.i.d. Random Vectors \cite{ref33,ref34}). Assume that $X_1^k, \ldots, X_n^k$ are $\mathbb{R}^k$-valued, zero-mean, i.i.d. random vectors with positive definite covariance $\operatorname{Cov}\left(X_1^k\right)$ and finite third absolute moment $\xi:=\mathrm{E}\left[\left\|X_1^k\right\|_2^3\right]$. Let $\mathbf{f}(\mathbf{x})$ be a vector-valued function from $\mathbb{R}^k$ to $\mathbb{R}^l$ that is also twice continuously differentiable in a neighborhood of $\mathbf{x}=\mathbf{0}$. Let $\mathbf{J} \in \mathbb{R}^{l \times k}$ be the Jacobian matrix of $\mathbf{f}(\mathbf{x})$ evaluated at $\mathbf{x}=\mathbf{0}$, i.e., its elements are
\begin{flalign*}
	J_{i j}=\left.\frac{\partial f_i(\mathbf{x})}{\partial x_j}\right|_{\mathbf{x}=\mathbf{0}},
\end{flalign*}
where $i=1, \ldots, l$ and $j=1, \ldots, k$. Then, for every $n \in \mathbb{N}$, we have
\begin{flalign*}
	\sup _{\mathscr{C} \in \mathfrak{C}_l}\left|\operatorname{Pr}\left(\mathrm{f}\left(\frac{1}{n} \sum_{i=1}^n X_i^k\right) \in \mathscr{C}\right)-\operatorname{Pr}\left(Z^l \in \mathscr{C}\right)\right| \leq \frac{c}{\sqrt{n}}
\end{flalign*}
where $c>0$ is a finite constant, and $Z^l$ is a Gaussian random vector in $\mathbb{R}^l$ with mean vector and covariance matrix respectively given as
\begin{flalign*}
	\mathrm{E}\left[Z^l\right]=\mathbf{f}(\mathbf{0}), \quad \text { and } \quad \operatorname{Cov}\left(Z^l\right)=\frac{\mathbf{J} \operatorname{Cov}\left(X_1^k\right) \mathbf{J}^{\prime}}{n} .
\end{flalign*}
\end{lemma}
	
	Denote for arbitrary $\kappa$
	\begin{flalign*}
		\log K=n\mathrm{C}(P)-\kappa\sqrt{n\mathrm{V}(P)}.
	\end{flalign*}
	
It is easily to check that for every $x^{n} \in \mathbb{R}^n$ such that $\|x^{n}\|_2^2=n P$, the moments of $\tilde{i}(x^{n};Y^{n})$ satisfy that
\begin{flalign*}
		\mathbb{E}\left[\frac{1}{n} \sum_{i=1}^n \log \frac{W\left(Y_i \mid x_i\right)}{Q_Y^*\left(Y_i\right)}\right]&=\mathrm{C}(P),\\
		\operatorname{Var}\left[\frac{1}{n} \sum_{i=1}^n \log \frac{W\left(Y_i \mid x_i\right)}{Q_Y^*\left(Y_i\right)}\right]&=\frac{\mathrm{V}(P)}{n},\\
		\mathbb{E}\big|\frac{1}{n} \sum_{i=1}^n \log \frac{W\left(Y_i \mid x_i\right)}{Q_Y^*\left(Y_i\right)}\big|^{3} &< +\infty.
\end{flalign*}

	According to Lemma 2, we have
	\begin{flalign}
		\Big| \mathrm{Pr}(\tilde{i}(x^{n};Y^{n})\leq \log K)-\mathrm{Q}(\kappa)\Big| \leq \frac{B}{\sqrt{n}},
		\label{eq15}
	\end{flalign}
	where $B>0$ is a finite constant.
	
	For sufficiently large $n$, let
	\begin{flalign*}
		\kappa=\mathrm{Q}^{-1}((1+\frac{2}{n})^{-2}\varepsilon-\frac{B}{\sqrt{n}}).%\tag{d}
	\end{flalign*}
	
	Then, from (\ref{eq15}) we obtain
	\begin{flalign}
		&\varepsilon_{n}\leq (1+\frac{2}{n})^{2}\mathrm{Pr}(\tilde{i}(x^{n};Y^{n})\leq \log K)\leq \varepsilon. \label{eq16}
	\end{flalign}
	
	Hence, 
	\begin{flalign}
		\log\log N^{*}(\varepsilon,\delta|W^{n})\geq& n\mathrm{C}(P)-\sqrt{n\mathrm{V}(P)}\mathrm{Q}^{-1}(\varepsilon)\nonumber\\
		&+\mathrm{O}(\log n),\quad \delta\rightarrow 0.\label{eq17}
	\end{flalign}
	\\\	
	\noindent \emph{Proof of converse.} In this part, we develop a finer quantization method depicted in Fig \ref{fig5}. 
	\begin{figure}[H]
		\centering
		\includegraphics[height=5.5cm,width=7.7cm]{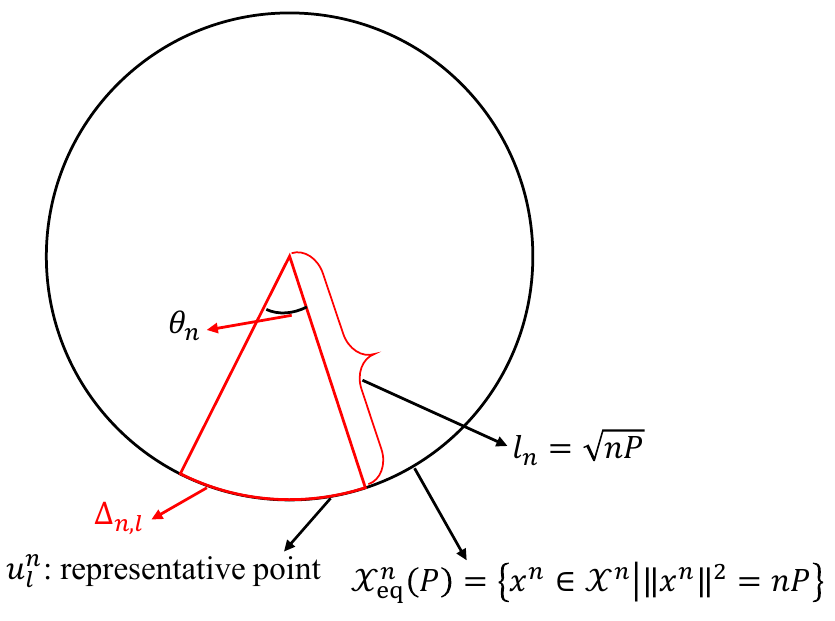}
		\caption{A finer quantization method.}
		\label{fig5}
	\end{figure}
	Let us provide some intuition for the quantization method above. First, we partition the power shell $\mathcal{X}_{\mathrm{eq}}^{n}(P)$ into small sectors $\Delta_{n,l} (l=1,\cdots,m_{n}(P))$ with angle $\theta_{n}$. %Next, partition $\Delta_{n,i}$ into small areas $\Delta_{n,i}^{(k)}$ with radius $\eta_{n}$. To be specific, choose a representative point $\mathbf{u_{k}}$ in  each $\Delta^{(k)}_{n,i}$, the total number of  $\Delta_{n,i}^{(k)}$ is $m_{n}(P)$, 
	Then, let
	\begin{flalign}
		\mathcal{V}_{n}(P)=\{u^{n}_{1},\cdots,u^{n}_{m_{n}(P)}\}\nonumber
	\end{flalign}
	and denote $\bar{Q}$ supported on $\mathcal{V}_{n}(P)$ by
	\begin{flalign}
		\bar{Q}(u^{n}_{l})=Q(\Delta_{n,l})  (l=1,2,...,m_{n}(P)),
		\tag{c}%\label{eq7}
	\end{flalign}
	in other words, $\bar{Q}$ is obtained by concentrating the mass of $Q$ in $\Delta_{n,l}$ to the representative point $u^{n}_{l}$.

	When the %radius $\eta_{n}$ and 
	angel $\theta_{n}$ is small, the output distribution induced by $\bar{Q}$ approximates the real output induced by $Q$ (see (\ref{eq3}) by replacing $\mathcal{X}^{n}(P)$ by $\mathcal{X}^{n}_{\mathrm{eq}}(P)$) well. %Unlike the hypercube covering quantization method provided in Fig 2, 
	Fig \ref{fig5} %directly partition the $n$-dimensional input alphabet $\mathcal{X}_{n}(P)$ and 
	proposes a finer quantization which leads to the second-order RID capacity.
	
	%We highlight some key points and give a short proof of the converse part.
	\begin{remark}
		Define an $n$-dimensional spherical coordinate system determined by angles $\phi_{1},\cdots,\phi_{n-1}$ and radius $l_{n}$,% as follows% some simple calculations show that  
		\begin{flalign*}
			\begin{cases}
				x_{1}=l_{n}\cos\phi_{1}\\
				x_{2}=l_{n}\sin\phi_{1}\cos\phi_{2}\\
				%x_{3}=l_{n}\sin\phi_{1}\sin\phi_{2}\\
				\cdots\\
				x_{n-2}=l_{n}\sin\phi_{1}\cdots\sin\phi_{n-2}\cos\phi_{n-1}\\
				x_{n-1}=l_{n}\sin\phi_{1}\cdots\sin\phi_{n-2}\sin\phi_{n-1}
			\end{cases}
		\end{flalign*}
		where $0\leq\phi_{m}\leq\pi(1\leq m\leq n-2)$, $0\leq\phi_{n-1}\leq 2\pi$.
	\end{remark}
	Then, the number of the sectors is
	\begin{flalign}
		m_{n}(P)=(\frac{\pi}{\theta_{n}})^{n-2}\frac{2\pi}{\theta_{n}}.\label{eq18}
	\end{flalign}
	%and the number of $\Delta_{n,i}^{(k)}$ is
	%		\begin{flalign}
		%	k_{n}(P)=(\frac{\pi}{\theta_{n}})^{n-2}\frac{2\pi}{\theta_{n}}\frac{l_{n}}{\eta_{n}}.\label{eq1111}
		%\end{flalign}
		\noindent$\mathbf{Step}$ $\mathbf{1.}$ Let $l_{n}=\sqrt{nP}$ be the radius of the power shell $\mathcal{X}_{\mathrm{eq}}^{n}(P)$ %$\eta_{n}=\exp(-n)$ 
		and substitute $\theta_{n}=\exp(-n)$ %$l_{n}$, $\theta_{n}$ %and $\theta_{n}$ 
		into $(\ref{eq18})$, we have
		\begin{flalign}
			\frac{\log\log m_{n}(P)}{\log n}\rightarrow C,\label{eq19}
		\end{flalign}
		where $C>0$ is a constant.
		
		Here, we give a brief proof of (\ref{eq19}).
		
		Substituing $\theta_{n}$ into $(\ref{eq18})$, so that $m_{n}(P)$ can be calculated as
		\begin{flalign}
			m_{n}(P)=\frac{2\pi^{n-1}}{(\exp(-n))^{n-1}}. \label{eq111}
		\end{flalign}
		
		The result shown in $(\ref{eq19})$ immediately follows by $(\ref{eq111})$.
		
		\noindent$\mathbf{Step}$ $\mathbf{2.}$ %(Evaluation of the TV distance) 
		Using the quantization distribution $\bar{Q}(\cdot)$ given by $\mathrm{(c)}$ and the Pinsker's inequality, we can show that
		\begin{flalign}
			d(QW^{n},\bar{Q}W^{n})\leq \sqrt{P}n^{3/2} e^{-n}.\label{eq20}
		\end{flalign}
		%where $c=\sqrt{P}$ is a constant. 
		
		Now, we give the proof of $(\ref{eq20})$. 
		
		For any $x^{n}, u^{n}_{l} \in \Delta_{n,l}$, from  the definition of the KL divergence between two multivariate Gaussian distributions, we obtain
		\begin{flalign}
			&D(W^{n}(\cdot|x^{n})||W^{n}(\cdot|u^{n}_{l}))= \frac{1}{2} n ||x^{n}-u^{n}_{l}||^{2}\nonumber\\
			&\leq \frac{n(n-1)}{2} (\sqrt{nP}e^{-n})^{2}. \label{eq21}%+\frac{n}{2}(\exp(-n))^{2}.%\nonumber\\
			%&=\frac{n(n-1)}{2}nP\exp(-2\sqrt{n})+\frac{n}{2}\exp(-2n).
		\end{flalign}
		
		By the Pinsker's inequality \cite{ref35} and the triangle inequality, we have
		\begin{flalign}
			&d(QW^{n},\bar{Q}W^{n})=2\sup_{\mathcal{B}\subset\mathbb{R}^{n}}|QW^{n}(\mathcal{B})-\bar{Q}W^{n}(\mathcal{B})| \nonumber\\
			&\leq 2\sum_{l=1}^{m_{n}(P)}\sum_{x^{n}\in\Delta_{n,l}} |W^{n}(\mathcal{B}|x^{n})-W^{n}(\mathcal{B}|u^{n}_{l})|Q(x^{n}) \nonumber\\
			&\leq  2\sum_{l=1}^{m_{n}(P)} \frac{1}{2}\sqrt{2D(W^{n}(\cdot|x^{n})||W^{n}(\cdot|u^{n}_{l}))}\bar{Q}(u^{n}_{l}) \nonumber\\
			&\leq \sqrt{P}n^{3/2}e^{-n}. \label{eq22}
		\end{flalign} 
		\noindent$\mathbf{Step}$ $\mathbf{3.}$ Next, we investigate the relation between  ID and  resolvability \cite{ref14, ref16, ref17, ref18, ref19, ref20}. We have the following lemma.
		\begin{lemma}For arbitrary $0<\varepsilon,\ \delta,\ \xi<1$ with $0<\varepsilon+\delta+\xi+\sqrt{P}n^{3/2}e^{-n}<1$, we have
			\begin{flalign}
				\log\log N^{*}(\varepsilon,\delta|W^{n})\leq \log M^{*}(\xi|W^{n})+\mathrm{O}(\log n),
				\label{eq23}
			\end{flalign}
			where $N^{*}(\varepsilon,\delta|W^{n})$ and $M^{*}(\xi|W^{n})$ are the optimal code size of the $(\varepsilon,\delta)$-$\mathrm{ID}$ code and the $\xi$-resolvability code, respectively.
		\end{lemma}
		\begin{proof}
			%The key point is the use of the Quantization method and type method. 
			%		(\ref{eq17}) is true as long as 
			%		\begin{flalign}
				%		\frac{1}{n} \log\log N_{n}(\varepsilon,\delta|W^{n})\leq \frac{1}{n}\log M_{n}(\xi|W^{n})+\mathrm{O}(\frac{\log{n}}{n})
				%			%\label{eq16}
				%		\end{flalign}
			%		is true.
			
			%		Suppose that $R_{1}$ and $R_{2}$ are $(\varepsilon,\delta)$-achievable ID rate and $\xi$-achievable resolvability rate, respectively. Then, 
			Let $N_{n}=N^{*}(\varepsilon,\delta|W^{n}), M_{n}=M^{*}(\xi|W^{n})$, by definition,
			there exist an $(n,N_{n},\varepsilon_{n},\delta_{n})$-ID code and $M_{n}$-type probability distributions $\tilde{Q}_{j} (j=1,\cdots,N_{n})$ supported on $\mathcal{X}_{\mathrm{eq}}^{n}(P)$ satisfying 
			\begin{flalign}
				&\varepsilon_{n}\leq\varepsilon,\quad \delta_{n}\leq\delta, \label{eq24}\\
				%			&\liminf_{n\rightarrow \infty} \frac{1}{n}\log\log N_{n}(\varepsilon,\delta|W^{n}) \geq R_{1},\\
				& d(Q_{j}W^{n},\tilde{Q}_{j}W^{n})\leq \xi, \label{eq25}
				%			& \limsup_{n\rightarrow \infty} \frac{1}{n}\log M_{n}(\xi|W^{n})\leq R_{2},
			\end{flalign}
			for all $j=1,\cdots,N_{n}$, where $Q_{j}$ is the random encoder.
			
			Let $\bar{\tilde{Q}}_{j} (j=1,\cdots,N_{n})$ be the quantization distribution of the $M_{n}$-type distribution $\tilde{Q}_{j}$. Therefore, $\bar{\tilde{Q}}_{j}$s are also $M_{n}$-types. By (\ref{eq20}), we have
			\begin{flalign}
				d(\tilde{Q}_{j}W^{n},\bar{\tilde{Q}}_{j}W^{n})\leq \sqrt{P}n^{3/2}e^{-n}.\label{eq26}% \nonumber\\
				%\forall j\in\{1,\cdots,N_{n}\}.
			\end{flalign} 
			
			Combine (\ref{eq25}) and (\ref{eq26}), we obtain% together tells us that
			\begin{flalign}
				d(Q_{j}W^{n},\bar{\tilde{Q}}_{j}W^{n})\leq\xi+\sqrt{P}n^{3/2}e^{-n}.\label{eq27}
			\end{flalign}
			%		where $\bar{\bar{Q}}_{j}$ is the quantization distribution of the $M_{n}$-type distribution $\tilde{Q}_{j}$.
			
			Suppose that there exists $j\neq k$, such that $\bar{\tilde{Q}}_{j}=\bar{\tilde{Q}}_{k}$, then, by the triangular inequality and the definition of the TV distance, 
			\begin{flalign}
				2(1-\varepsilon-\delta)\leq d(Q_{j}W^{n},Q_{k}W^{n})\leq 2\xi+2\sqrt{P}n^{3/2}e^{-n},\nonumber%\label{eq29}
			\end{flalign}
			i.e., $\varepsilon+\delta+\xi+\sqrt{P}n^{3/2}e^{-n}\geq1$, which is a contradiction!		
			%		Contradicting to the assumption in lemma 1.
			
			Therefore, $\bar{\tilde{Q}}_{j} (j=1,\cdots,N_{n})$ must be distinct $M_{n}$-type distributions, since the number of distinct $M_{n}$-types is upper bounded by $m_{n}(P)^{M_{n}}$. Hence, it must hold that			 $N_{n}\leq m_{n}(P)^{M_{n}}$,
			which implies
			%Multiplying $\frac{1}{n}$ on both sides yields
			\begin{flalign}
				\log\log N_{n}\leq \log M_{n}+\log\log m_{n}(P).\label{eq28}
			\end{flalign}
			
			By (\ref{eq19}) %leads to
			\begin{flalign}
				\log\log N_{n}\leq \log M_{n}+\mathrm{O}(\log n).\label{eq29}
			\end{flalign}
			
			This completes the proof of (\ref{eq23}).
			%Complete version was given in [16] due to space limitations. 
		\end{proof}

		A channel $W=\{(\mathcal{X},\mathcal{F}),(\mathcal{Y},\mathcal{G}),K\}$ is defined by an input alphabet $\mathcal{X}$ with $\sigma$-algebra $\mathcal{F}$, an output alphabet $\mathcal{Y}$ with $\sigma$-algebra $\mathcal{G}$, and a stochastic kernel $K$ that specifies the stochastic transition between the input and output alphabets. Subsequently, the lemma presented below provides a second-order achievability bound for the channel resolvability capacity, paving the way for establishing the second-order RID capacity for AWGNCs.
		\begin{lemma} (Frey's Theorem 3  \cite{ref36}) Given a channel $W=\{(\mathcal{X},\mathcal{F}),(\mathcal{Y},\mathcal{G}),K\}$ and an input distribution $Q_{X}$ such that the information density $i(X;Y)$ has finite central second moment $V$ and finite absolute third moment $\rho$, $\xi>0$ and $c>1$, suppose the rate $R$ depends on $n$ in the following way:
		\begin{flalign*}
			R=I(X;Y)+\sqrt{\frac{V}{n}}\mathrm{Q}^{-1}(\xi)+c\frac{\log n}{n}.
		\end{flalign*}
		Then, for any $d\in (0,c-1)$ and $n$ that satisfy $n^{(c-d)/2}\geq 6$, we have
		\begin{flalign*}
			P_{\mathcal{C}}&(||	P_{Y^{n}[c^{M}]}-QW^{n}||_{2}\geq \mu(1+\frac{1}{\sqrt{n}})+\frac{1}{\sqrt{n}})\\
			&\leq \exp(-\frac{1}{3}n\mu \exp(nR))\\
			&+(\frac{7}{6}+\sqrt{3\pi/2}\exp(\frac{3}{4}))\exp(-n^{\frac{1}{2}(c-d-1)}),
		\end{flalign*}			
		where
		\begin{flalign*}
			\mu:= \mathrm{Q}(\mathrm{Q}^{-1}(\xi)+d\frac{\log n}{\sqrt{nV}})+\frac{\rho}{V^{\frac{3}{2}}\sqrt{n}}
		\end{flalign*}
		tends to $\xi$ for $n\rightarrow\infty$ and $||x||_{2}$ denote the $\ell^{2}$-norm of $x$.% is denoted by .%, $\mathcal{C}=(C(1),\cdots,C(\exp(nR)))$ is the codebook where $C(m)\sim Q_{X}, (m=1,\cdots,\exp(nR))$ are the codewords.
		\end{lemma}
		
		\begin{remark}
			Given that the aforementioned lemma applies to any channel, when we consider the channel as an AWGNC, we can derive the converse bound for the second-order RID capacity of the AWGNCs.
		\end{remark}
		
		\noindent$\mathbf{Step}$ $\mathbf{4.}$ 		
		Let
		\begin{flalign}
			\xi=1-\varepsilon-\delta-\sqrt{P}n^{3/2}e^{-n}.
			\label{eq30}
		\end{flalign}
		
		Let $\delta\rightarrow 0$, i.e., the type-II error vanishes, from (\ref{eq23}) and Lemma 4, we obtain
		\begin{flalign}
			&\log\log N^{*}(\varepsilon,\delta|W^{n})\leq \log M^{*}(\xi|W^{n})+\mathrm{O}(\log n)\nonumber\\
			&\leq  n\mathrm{C}(P)+ \sqrt{n\mathrm{V}(P)}\mathrm{Q^{-1}}(1-\varepsilon-\sqrt{P}n^{3/2}e^{-n})
			+\mathrm{O}(\log n)\nonumber\\
			&=n\mathrm{C}(P)-\sqrt{n\mathrm{V}(P)}\mathrm{Q}^{-1}(\varepsilon+\sqrt{P}n^{3/2}e^{-n})+\mathrm{O}(\log n)\nonumber\\
			&\leq n\mathrm{C}(P)-\sqrt{n\mathrm{V}(P)}\mathrm{Q}^{-1}(\varepsilon)+\mathrm{O}(\log n),
			\label{eq31}
		\end{flalign}
		where the last inequality is due to the Taylor expansion %for the $\mathrm{Q}^{-1}$ function
		\begin{flalign*}
			&\mathrm{Q}^{-1}(\varepsilon+\sqrt{P}n^{3/2}e^{-n}) \\
			&=\mathrm{Q}^{-1}(\varepsilon)-\underbrace{\left.\frac{d\mathrm{Q}^{-1}(x)}{dx}\right|_{x=\varepsilon}}_{\mathrm{O}(1)}\sqrt{P}n^{3/2}e^{-n}+\mathrm{o}(n^{3/2}e^{-n}) \\
			&\geq \mathrm{Q}^{-1}(\varepsilon)+\mathrm{o}(\frac{1}{\sqrt{n}}).
		\end{flalign*}
		
		Thus, we prove the converse part.
		
		%Combine (\ref{eq1l}) and (\ref{eq13}), we prove 
		The desired results follow from (\ref{eq17}) and (\ref{eq31}).\qed
		%Hence, the proof of Theorem 2 is established.
		%\section{Simulation Results}
		%In this section, we demonstrate the advantages of our approximation results (\ref{eq8}) over the finite block length by some numerical experiments.	We plot the rate-blocklength trade-off in the RID framework for the AWGN channels with $\mathrm{SNR}=0 \mathrm{dB}$ and $\mathrm{SNR}=20 \mathrm{dB}$ where the type-I errors are bounded by $\varepsilon=10^{-3}$ and $\varepsilon=10^{-6}$, respectively.
		%
		%\begin{figure}[H]
		%	\centering
		%	\includegraphics[height=7cm,width=9cm]{RID_rate_blocklength.png}
		%	\caption{RID rate $\frac{\log\log N^{*}(\varepsilon,\delta|W^{n})}{n}$ vs block length $n$ for the AWGN channels.}
		%	\label{fig6}
		%\end{figure}
		%		
		%Comparing to the first-order RID capacity provided by Han \cite{ref10}, our results is more accurate and also quite competitive for finite $n$.		
		
		\section{Conclusion}
		The second-order RID capacity of the AWGNC has been investigated in this paper. It has been proved that the second-order RID capacity have the same form as the  transmission capacity in AWGNC. An extension of Hayashi's Theorem to the auxiliary output distribution leads to the achievability part and a finer quantization method has been proposed in the proof of the converse part which gives the $\mathrm{O}(\log n)$ term. As a future direction, it is tempting to apply the quantization approach to other identification problems. Furthermore, it remains an open problem to determine the second-order $(\varepsilon,\delta)$-ID capacity of the AWGNC for non-vanishing type-II error $\delta$, i.e., $\delta \geq \varepsilon_{0}$ with $\varepsilon_{0}>0$ a positive constant. It is also an interesting work to derive the third-order $(\varepsilon,\delta)$-ID capacity of the AWGNC.

		\section*{Acknowledgments}
		The authors thank the anonymous referees for their valuable comments. This work is supported by National Key R\&D Program of China No. 2023YFA1009601 and 2023YFA1009602
		%This should be a simple paragraph before the References to thank those individuals and institutions who have supported your work on this article.
		
		%\section{References Section}
		%You can use a bibliography generated by BibTeX as a .bbl file.
		%BibTeX documentation can be easily obtained at:
		%http://mirror.ctan.org/biblio/bibtex/contrib/doc/
		%The IEEEtran BibTeX style support page is:
		%http://www.michaelshell.org/tex/ieeetran/bibtex/
		
		% argument is your BibTeX string definitions and bibliography database(s)
		%\bibliography{IEEEabrv,../bib/paper}

	\end{document}